\documentclass[11pt,letterpaper]{article}

\usepackage[utf8]{inputenc}
\usepackage[english]{babel}

\usepackage{subcaption}

\usepackage[margin=1in]{geometry}
\usepackage{amsfonts,amsmath,amsthm,thmtools}
\usepackage{nicefrac}

\usepackage[dvipsnames,usenames]{xcolor}
\usepackage[colorlinks=true,pdfpagemode=UseNone,urlcolor=RoyalBlue,linkcolor=RoyalBlue,citecolor=OliveGreen,pdfstartview=FitH]{hyperref}
\usepackage[capitalize,noabbrev]{cleveref}

\declaretheorem{theorem}
\declaretheorem[sibling=theorem]{lemma}
\declaretheorem[sibling=theorem]{corollary}

\usepackage{todonotes}

\usepackage{booktabs}
\usepackage{caption}
\usepackage[numbers, sort]{natbib}

\usepackage{tcolorbox}

\usepackage{enumitem}

\usepackage{algorithm}
\usepackage{algpseudocode}
\algtext*{EndWhile}
\algtext*{EndIf}
\algtext*{EndFor}

\newcommand{\E}{\mathbf{E}}

\newcommand{\Sref}[1]{\hyperref[#1]{\S\ref*{#1}}}

\newcommand{\p}[1]{\left( #1 \right)}

\newcommand{\eps}{\varepsilon}

\title{Lossless Derandomization for Undirected Single-Source Shortest Paths and Approximate Distance Oracles}
\author{
    Shuyi Yan
    \thanks{BARC, University of Copenhagen. Supported by VILLUM Foundation Grant 54451. Email: \texttt{shya@di.ku.dk}.}
}
\date{}

\begin{document}

\begin{titlepage}
    \thispagestyle{empty}
    \maketitle
    \begin{abstract}
        \thispagestyle{empty}
        A common step in algorithms related to shortest paths in undirected graphs is that, we select a subset of vertices as centers, then grow a ball around each vertex until a center is reached.
We want the balls to be as small as possible.
A randomized algorithm can uniformly sample $r$ centers to achieve the optimal (expected) ball size of $\Theta(n/r)$.
A folklore derandomization is to use the $O(\log n)$ approximation for the set cover problem in the hitting set version where we want to hit all the balls with the centers.

However, the extra $O(\log n)$ factor is sometimes too expensive. For example, the recent $O(m\sqrt{\log n\log\log n})$ undirected single-source shortest path algorithm \cite{duan2023randomized} beats Dijkstra's algorithm in sparse graphs, but the folklore derandomization would make it dominated by Dijkstra's.

In this paper, we exploit the fact that the sizes of these balls can be adaptively chosen by the algorithm instead of fixed by the input.
We propose a simple deterministic algorithm achieving the optimal ball size of $\Theta(n/r)$ on average. Furthermore, given any polynomially large cost function of the ball size, we can still achieve the optimal cost on average.
It allows us to derandomize \cite{duan2023randomized}, resulting in a deterministic $O(m\sqrt{\log n\log\log n})$ algorithm for undirected single-source shortest path.

In addition, we show that the same technique can also be used to derandomize the seminal Thorup-Zwick approximate distance oracle \cite{thorup2005approximate}, also without any loss in the time/space complexity.
    \end{abstract}
\end{titlepage}

\section{Introduction}

The single-source shortest path is a fundamental problem in graph theory. Given a graph with $n$ vertices and $m$ edges with non-negative real weights, we want to find the distances from a given source $s$ to all vertices. The textbook Dijkstra's algorithm \cite{dijkstra1959note} with a Fibonacci heap \cite{fredman1987fibonacci} can solve this problem in $O(m+n\log n)$ time.
Recently, \cite{duan2023randomized} gave the first algorithm breaking this bound for undirected sparse graphs. Their algorithm is randomized and runs in $O(m\sqrt{\log n\log\log n})$ time.

Their algorithm uses randomization to solve the following subproblem, which widely appears in algorithms regarding shortest paths, distance oracles and spanners. We want to select a subset of vertices $R$ as centers. Then for each vertex $v$, we grow a ball $B(v)$ until we reach a center, i.e. $B(v)$ contains all vertices that are closer than (all vertices in) $R$ to $v$.
The performance of the algorithm then depends on the size of $R$ and the sizes of the balls. As an example, for \cite{duan2023randomized}, the total time complexity will be $O(|R|\log n + \sum_{v}|B(v)|\log|B(v)|)$\footnote{This is for their main case that $n=\Omega(m)$.}.

For a randomized algorithm, it's easy to achieve the optimal tradeoff between these two sizes. By sampling $r$ centers uniformly at random, the expected size of each ball will be $\Theta(n/r)$\footnote{This is optimal, for example, when the graph is a path.}. There is a folklore derandomization, which first fixes all the balls and then repeatedly selects a new center which hits the largest number of unhit balls.
It introduces an additional $O(\log n)$ factor on the number of centers, and some algorithms cannot afford this. For example, applying this derandomization to \cite{duan2023randomized} will increase the time complexity to $O(m\log n\sqrt{\log\log n})$.

However, it is unnecessary to fix the balls in the beginning. In this paper, we exploit the fact that these balls are growable by the algorithm. We show a very simple derandomization which preserves the optimal randomized bound in an amortized way. The result can be generalized to any polynomially large function of the ball size.
Below, we state our result in the context of derandomizing \cite{duan2023randomized}. The formal and more general results are included in \cref{sec:hitting-balls}.

\begin{theorem}[Informal]
\label{thm:hitting-balls-informal}
    For any $1\le r\le n$, we can deterministically select the center set $R$ with size $|R|=r$, such that $\sum_v|B(v)|\log|B(v)|=O(n\cdot(n/r)\log(n/r))=O((n^2/r)\log(n/r))$, in $O(r+(n^2/r)\log(n/r))$ time.
\end{theorem}

Combining \cref{thm:hitting-balls-informal} with the remaining deterministic part of \cite{duan2023randomized}, we get a deterministic algorithm with the same time complexity. The details are included in \cref{sec:sssp}.

\begin{theorem}
\label{thm:sssp}
    In an undirected graph with non-negative edge weights, there is a deterministic comparison-addition algorithm that solves the single-source shortest path problem in $O(m\sqrt{\log n\log\log n})$ time.
\end{theorem}

Our derandomization can also be applied to some other randomized algorithms that involve similar center-ball techniques.
An example is the seminal Thorup-Zwick approximate distance oracle \cite{thorup2005approximate}. For any integer $k\ge 1$, \cite{thorup2005approximate} can preprocess an undirected non-negatively weighted graph in $O(kn^{1/k}(m+n\log n))$ expected time, constructing a data structure of size $O(kn^{1+1/k})$, which is able to answer any pairwise distance query, of stretch $2k-1$\footnote{It means that, if the real distance between $u$ and $v$ is $d(u,v)$, then the estimated distance $\hat{d}(u,v)$ satisfies $d(u,v)\le\hat{d}(u,v)\le(2k-1)d(u,v)$.}, in $O(k)$ time.

The original derandomization in \cite{thorup2005approximate} takes $\tilde{O}(mn)$ time and increases the space complexity to $O(kn^{1+1/k}\log n)$. Later, \cite{roditty2005deterministic} proposed a better derandomization, which only loses one $O(\log n)$ factor in the time complexity and has no (asymptotic) loss in the space complexity. However, using our method, this last $O(\log n)$ factor can also be removed. The details are included in \cref{sec:do}.

\begin{theorem}
\label{thm:do}
    The Thorup-Zwick approximate distance oracle of size $O(kn^{1+1/k})$ can be deterministically constructed in $O(kn^{1/k}(m+n\log n))$ time.
\end{theorem}

We remark that \cref{thm:do} can be combined with other improved query algorithms for the Thorup-Zwick oracle. For example, combining it with \cite{wulff2013approximate} will give us $O(\log k)$ query time.

\subsection{Other Related Works}

\paragraph{Single-Source Shortest Path.} Very recently, \cite{DBLP:conf/stoc/DuanMMSY25} proposed an $O(m\log^{2/3}n)$-time deterministic algorithm for directed graphs, using different techniques from \cite{duan2023randomized}.

In the word RAM model with integer weights, there are a sequence of works \cite{DBLP:journals/jcss/FredmanW93,DBLP:journals/jcss/FredmanW94,DBLP:conf/esa/Raman96,DBLP:journals/sigact/Raman97,DBLP:journals/jal/Thorup00,DBLP:journals/siamcomp/Thorup00,DBLP:conf/icalp/Hagerup00} leading to a linear-time algorithm for undirected graphs \cite{DBLP:journals/jacm/Thorup99} and an $O(m+n\log\log\min\{n,C\})$-time algorithm for directed graphs \cite{DBLP:journals/jcss/Thorup04} where $C$ is the maximum edge weight.

\paragraph{Approximate Distance Oracle.} The size-stretch tradeoff of the Thorup-Zwick oracle \cite{thorup2005approximate} is optimal up to a factor of $k$, assuming a widely believed girth conjecture \cite{erdos1964extremal}. Besides, it is the first approximate distance ``oracle'' that achieves a constant query time for any fixed stretch. Before that, the query time was $\Omega(n^{1/k})$ \cite{DBLP:journals/dcg/AlthoferDDJS93,DBLP:journals/siamcomp/AwerbuchBCP98,DBLP:journals/siamcomp/Cohen98,matouvsek1996distortion}.

The query time of the Thorup-Zwick oracle was improved to $O(\log k)$ by \cite{wulff2013approximate}.
A line of research achieved a universal constant query time and $O(n^{1+1/k})$ size \cite{DBLP:conf/focs/MendelN06,DBLP:journals/cjtcs/MendelS09,wulff2013approximate,DBLP:conf/stoc/Chechik14,DBLP:conf/stoc/Chechik15} by designing other approximate distance oracles. However, they have worse preprocessing times and/or worse stretches.

There are also many further improvements in various special settings \cite{DBLP:journals/talg/BaswanaS06,DBLP:conf/focs/BaswanaK06,DBLP:conf/icalp/BaswanaGSU08,DBLP:journals/siamcomp/BaswanaK10,DBLP:conf/soda/Wulff-Nilsen12,DBLP:journals/siamcomp/PatrascuR14,DBLP:conf/focs/PatrascuRT12,DBLP:conf/soda/AgarwalG13,DBLP:conf/esa/Agarwal14,roditty2023approximate,DBLP:conf/soda/ChechikZ22,DBLP:conf/icalp/KopelowitzKR24}, e.g. for sparse/dense graphs, special stretches, etc.
\section{Hitting Growable Balls}
\label{sec:hitting-balls}

We define the problem of hitting growable balls as follows.
There are $n$ vertices (elements), $m$ balls (sets) which are initially empty, a parameter $r\in[1,n]$ and a cost function $f(\cdot)$. The algorithm can select a ball (which does not contain all vertices) and ask the adversary to include one new vertex (determined by the adversary) in it. The algorithm can grow the balls as many times as it wants. In the end, it needs to select $r$ vertices as centers, such that every ball is hit by the centers, i.e. each ball contains at least one center.
The objective is to minimize the total cost $\sum_{i=1}^m f(|B_i|)$, where $B_i$ denotes the $i$-th ball.

The cost function is typically positive and convex.
For simplicity of presentation, in the remaining part of this section, we assume that the cost function is $f(x)=x^p$ for some constant $p\ge 1$. Our results could be extended to other polynomially large cost functions simply by applying Jensen’s inequality.

A randomized algorithm could simply select the centers at random, and then grow each ball until it contains a center. Against an oblivious adversary, the expected cost will be $O_p(m\cdot f(n/r))$\footnote{$O_p(\cdot)$ indicates that the constant hidden in the $O$-notation may depend on $p$.}, which is (asymptotically) optimal since the adversary can easily force the average ball size to be $\Omega(n/r)$\footnote{An example is that, when $B_i$ grows in the $j$-th time, the adversary let it include the vertex $v_{(i+j)\ \text{mod}\ n}$.}.

Surprisingly, we have a simple deterministic algorithm achieving the same bound. The idea is as follows.
When each unhit ball has a size of $n/r$, we can use $O(r)$ centers to hit a constant fraction, say $1-\eps$, of unhit balls, by repeatedly selecting a new center that hits the largest number of unhit balls. Repeating it $O(\log m)$ times, we can hit all the balls with $O(r\log m)$ centers. To remove the $O(\log m)$ factor, we exponentially reduce the number of new centers in each round. In the $i$-th round, we only use $O(r/2^i)$ centers. This means that each unhit ball should have size $2^in/r$. However, only $\eps^i$ fraction of the balls remain unhit in the $i$-th round. By setting $\eps<1/2$, the average ball size would still be $O(n/r)$.

The pseudocode of our algorithm is given below. Recall that the objective is to minimize the total cost $\sum_{i=1}^m f(|B_i|)$ where the cost function is $f(x)=x^p$.

\begin{algorithm}[H]
    \caption{}
    \label{alg:hit-ball}
    \hspace*{\algorithmicindent} \textbf{Input:} number of vertices $n$, number of balls $m$, target number of centers $r$, cost parameter $p$. \\
    \hspace*{\algorithmicindent} \textbf{Output:} at most $r$ centers that hit all the balls.
    \begin{algorithmic}[1]
        \State Initially, all balls are empty and no centers are selected.
        \State $b\gets \lceil 2^{p+2}n/r \rceil$.
        \State Grow each ball to the size $\min\{b,n\}$.
        \While{not all the balls are hit}
            \State $m'\gets$ the number of unhit balls.
            \State Repeatedly select a new center that hits the largest number of unhit balls, until the number of unhit balls is at most $m'/2^{p+1}$.
            \State $b\gets 2b$.
            \State Grow each unhit ball to the size $\min\{b,n\}$.
        \EndWhile
        \State \Return all selected centers.
    \end{algorithmic}
\end{algorithm}

\begin{theorem}
\label{thm:hitting-ball}
    For any $p\ge1$ and $1\le r\le n$, \cref{alg:hit-ball} selects at most $r$ centers that hit all the balls at a total cost $O_p(m(n/r)^p)=O_p(m\cdot f(n/r))$, and it can be implemented in $O_p(r+mn/r)$ time.
\end{theorem}

\begin{proof}
    As long as $b\ge n$, the algorithm will terminate in one step. Before that, there are at most $\log r$ rounds.
    In the $i$-th round, each unhit ball has a size of at least $2^{p+1+i}n/r$. When we select a new center, there are at least $m'/2^{p+1}$ unhit balls by definition. So, a new center selected in the $i$-th round will hit at least
    $$\frac{m'}{2^{p+1}}\cdot\frac{2^{p+1+i}n}{r}\cdot\frac1n = \frac{2^im'}{r}$$
    unhit balls. On the other hand, there are only $m'$ unhit balls. So, at most $r/2^i$ new centers are selected in the $i$-th round, which means the total number of selected centers is at most $r$.

    In each round, the number of unhit balls is divided by at least $2^{p+1}$. So, at the beginning of the $i$-th round, at most $m/2^{(i-1)(p+1)}$ balls are unhit, which have grown to the size $2^{p+1+i}n/r$. Therefore, the total cost is at most
    \begin{align*}
        \sum_{i=1}^{\log r} \frac{m}{2^{(i-1)(p+1)}}\cdot f\p{\frac{2^{p+1+i}n}{r}}
        & = \sum_{i=1}^{\log r} \frac{m}{2^{(i-1)(p+1)}}\cdot \p{\frac{2^{p+1+i}n}{r}}^p \\
        & = m(n/r)^p \sum_{i=1}^{\log r} 2^{(p+1)^2-i} \\
        & \le 2^{(p+1)^2}m(n/r)^p \\
        & = O_p(m(n/r)^p).
    \end{align*}

    Let $a_j$ denote how many unhit balls can be hit by vertex $j$. When we grow an unhit ball, either it becomes hit by some current center, or some $a_j$ is increased by $1$. When we hit a ball $B_i$, we decrease $|B_i|$ different $a_j$'s by $1$ since $B_i$ is no longer unhit. Therefore, the total number of increases/decreases is $\sum_{i=1}^m|B_i|=O_p(mn/r)$. Since $a_j$ is only changed by $1$ each time, each increase, decrease and find-max operation can be done in constant time\footnote{For example, we can maintain: (1) $a_j$ for each vertex $j$; (2) a doubly linked list $L_k$ of vertices with $a_j=k$ for each $k$; (3) the largest $k$ such that $L_k$ is non-empty.}. Therefore, the time complexity is $O_p(r+mn/r)$.
\end{proof}

In the applications, we will set the cost function according to the time complexity of the remaining part of the algorithm. In this paper, it is either $f(x)=x$ or $f(x)=x\log x$.
By Jensen's inequality, for any $p>1$, $\frac1m\sum_{i=1}^m|B_i|^p\le(n/r)^p$ would imply $\frac1m\sum_{i=1}^m|B_i|\log|B_i|\le(n/r)\log(n/r)$.
This gives us the following corollary.

\begin{corollary}
\label{cor:hit}
    Let the cost function be $f(x)=x\log x$. For any $1\le r\le n$, \cref{alg:hit-ball} (with e.g. $p=2$) selects at most $r$ centers that hit all the balls at a total cost $O(m(n/r)\log(n/r))=O(m\cdot f(n/r))$ in $O(r+mn/r)$ time.
\end{corollary}

\begin{proof}
    Let $g(x)=\sqrt{x}\log\sqrt{x}$ and $h(x)=x^2$. Note that $g$ is increasing and concave (when $x\ge1$), and $f(x)=g(h(x))$. Running \cref{alg:hit-ball} with $p=2$, we have $$\frac{1}{m}\sum_{i=1}^m h(|B_i|) \le h(n/r).$$
    So, $$\frac{1}{m}\sum_{i=1}^{m} f(|B_i|) = \frac{1}{m}\sum_{i=1}^{m} g(h(|B_i|)) \le g\left(\frac{1}{m}\sum_{i=1}^m h(|B_i|)\right) \le g(h(n/r)) = f(n/r).$$
\end{proof}

We remark that the result can be similarly extended to other polynomially large cost functions.
\section{Preliminaries for the Applications}

In the remaining part of the paper, we consider a weighted undirected graph $G$ with $n$ vertices, $m$ edges and non-negative edge weights.
For each edge $(u,v)$, let $l(u,v)$ denote the weight (length) of the edge.
For any vertices $u$ and $v$, let $d(u,v)$ denote the distance between $u$ and $v$.
For any vertex $u$ and any set of vertices $S$, let $d(u,S)=\min_{v\in S}\{d(u,v)\}$.

\paragraph{Bounded-Degree Graph.} We assume that the degree of each vertex is $O(m/n)$. This can be accomplished by splitting the vertices of too large degrees so that in the new graph, the number of vertices is still $O(n)$ and the number of edges is still $O(m)$.

\paragraph{Comparison-Addition Model.} We consider the comparison-addition model, where the weights are subject to only comparison and addition operations. Each comparison and addition takes unit time.

\paragraph{Heap.} We will use heaps that support each initialization, insertion and decrease-key operation in $O(1)$ amortized time and each extract-min operation in $O(\log|H|)$ time where $|H|$ is the size of the heap, e.g. the Fibonacci heap \cite{fredman1987fibonacci}.

\paragraph{Dijkstra's Algorithm.} The Dijkstra's algorithm \cite{dijkstra1959note} starts at a source vertex $s$ and repeatedly finds the closest unvisited vertex from a heap. In the $i$-th step, the heap only contains the neighbors of the $i$ closest vertices, hence its size is at most $im/n$. Therefore, it spends $O(m/n+\log(im/n))$ time in the $i$-th step.

\section{Derandomizing Undirected Single-Source Shortest Paths}
\label{sec:sssp}

\subsection{Bundle Dijkstra}

In the single-source shortest path problem, given a source vertex $s$, we need to compute $d(s,v)$ for each vertex $v$. The randomized algorithm in \cite{duan2023randomized} solves the problem in $O(m\sqrt{\log n\log\log n})$ time. The algorithm consists of two stages: bundle construction and bundle Dijkstra.

\paragraph{Constant-Degree Graph.} For simplicity, \cite{duan2023randomized} further assumes that each vertex has a constant degree. To achieve this, the number of vertices is increased to $O(m)$. We will follow the same assumption.
Note that this assumption is not necessary for their algorithm (See \cite{duan2023randomized} Section 5). Without it, the time complexity can be slightly improved to $O(\sqrt{mn\log n}+n\sqrt{\log n\log\log n})$ when $m=\omega(n)$ and $m=o(n\log n)$, and our derandomization also holds.

\paragraph{Bundle Construction.} In this stage, we need to select a subset of vertices $R$ as centers. Then for each vertex $v$, we run a partial Dijkstra's algorithm from $v$ until we reach a center. Let $B(v)$ be the set of vertices that are reached before the center during the partial Dijkstra's algorithm. The time complexity is $O(\sum_v|B(v)|\log|B(v)|)$ on a constant-degree graph. \cite{duan2023randomized} selects the centers at random, and this is the only part we need to derandomize.

\paragraph{Bundle Dijkstra.} In this stage, given $R$ and all $B(v)$, the bundle Dijkstra algorithm (\cite{duan2023randomized} Section 3.2) can solve the single-source shortest path problem in $O(|R|\log n+\sum_v|B(v)|)$ time deterministically.

~

In total, the time complexity is $O(|R|\log n+\sum_v|B(v)|\log|B(v)|)$. When the centers are selected uniformly at random, we have $\E[|B(v)|\log|B(v)|]=(m/|R|)\log(m/|R|)$. Setting $|R|=m\sqrt{\log\log n}/\sqrt{\log n}$ gives us the time complexity $O(m\sqrt{\log n\log\log n})$.

\subsection{Derandomization}

The bundle construction can be modeled by our hitting-growable-balls problem in \cref{sec:hitting-balls}. Each $B(v)$ corresponds to one ball, which is initially empty. When \cref{alg:hit-ball} asks to grow $B(v)$, we take one more step in the partial Dijkstra from $v$, which will include one more vertex in $B(v)$.
In the end, \cref{alg:hit-ball} selects $r$ centers that hit all the balls $B(v)$, which means each partial Dijkstra has reached a center.

The total time complexity for growing the balls is $O(\sum_v|B(v)|\log|B(v)|)$, the same as the original bundle construction. So we let the cost function be $f(x)=x\log x$. By \cref{cor:hit}, we have $\sum_v|B(v)|\log|B(v)|=O(m(n/r)\log(n/r))$.
Replacing the randomized bundle construction in \cite{duan2023randomized} by this deterministic bundle construction and setting $r=m\sqrt{\log\log n}/\sqrt{\log n}$, we get a deterministic algorithm with the same time complexity $O(m\sqrt{\log n\log\log n})$, proving \cref{thm:sssp}.
\section{Derandomizing Approximate Distance Oracles}
\label{sec:do}

\subsection{Thorup-Zwick Oracle}

An approximate distance oracle is a data structure which is able to approximately answer any pairwise distance query. If the estimated distance $\hat{d}(u,v)$ returned by the oracle always satisfies $d(u,v)\le\hat{d}(u,v)\le k\cdot d(u,v)$, we say that the oracle is of stretch $k$.

The Thorup-Zwick approximate distance oracle \cite{thorup2005approximate} works as follows.
Let $A_0,A_1,\dots,A_k$ be sets of vertices such that $V=A_0\supseteq A_1\supseteq\cdots\supseteq A_k=\emptyset$. Consider each level $i\in[0,k-1]$. For each vertex $v$, let $B_i(v)$ be the set of vertices in $A_i$ which are closer to $v$ than any vertex in $A_{i+1}$, i.e. $B_i(v)=\{u\in A_i \mid d(v,u)<d(v,A_{i+1})\}$\footnote{In particular, we define $d(v,A_k)=d(v,\emptyset)=\infty$.}. The Thorup-Zwick oracle stores all $B_i(v)$ and has size $O(\sum_{i=0}^{k-1}\sum_v|B_i(v)|)$. \cite{thorup2005approximate} show that it can answer each query of stretch $2k-1$ in $O(k)$ time.

Since the query part is already deterministic, we focus on the preprocessing time and the size. \cite{thorup2005approximate} let $A_{i+1}$ contain each element in $A_i$ with probability $n^{-1/k}$, so that each $B_i(v)$ has an expected size of $O(n^{1/k})$. Therefore the size of the data structure is $O(kn^{1+1/k})$. In each level, to compute all $B_i(v)$, they run a partial Dijkstra's algorithm from each vertex in $A_{i+1}$, so that only the vertices in each $B_i(v)$ are visited. The preprocessing time is therefore $O(kn^{1+1/k}(m/n+\log n))=O(kn^{1/k}(m+n\log n))$.
However, this method requires us to know the entire $A_{i+1}$ before computing the balls, so it does not work well with our derandomization method.

\subsection{Derandomization}

Consider each level $i\in[0,k-2]$. Suppose that $A_0,\dots,A_i$ have been determined and we are going to select $n^{-1/k}|A_i|$ vertices from $A_i$ as $A_{i+1}$. This can be modeled by our hitting-growable-balls problem in \cref{sec:hitting-balls}. Each $B_i(v)$ corresponds to one ball, which is initally empty. However, now the set of vertices (in this hitting-growable-balls problem) is $A_i$ instead of the vertex set $V$ of the graph, which means it is non-trivial to grow a ball quickly.

In fact, the previous derandomization \cite{roditty2005deterministic} faced the same problem. They proposed an algorithm, which, from our point of view, can grow all the balls (including those that have been hit) together in the desired time.

Let $S_{i,j}(v)$ denote the (ordered) set of $j$ vertices in $A_i$ that are closest to $v$. Ties are broken by some identical total order of vertices, e.g. by the indices. $S_{i,1}(v)$ for all $v$ can be easily computed in $O(m+n\log n)$ time. Then the following algorithm can compute all $S_{i,j+1}(v)$ from all $S_{i,j}(v)$.

\begin{algorithm}[H]
    \caption{(restate from the proof of \cite{roditty2005deterministic} Theorem 2)}
    \label{alg:grow-ball}
    \hspace*{\algorithmicindent} \textbf{Input:} undirected graph $G=(V,E)$, step $j$, sets $S_{i,j}(v)$ for all $v\in V$. \\
    \hspace*{\algorithmicindent} \textbf{Output:} sets $S_{i,j+1}(v)$ for all $v\in V$.
    \begin{algorithmic}[1]
        \State Add a new source vertex $s$ to $G$.
        \For{each edge $(u,v)\in E$ in each direction}
            \If{$S_{i,j}(u)=S_{i,j}(v)$}
                \State Keep the edge $(u,v)$.
            \Else
                \State Let $w$ denote the first vertex in $S_{i,j}(u)\setminus S_{i,j}(v)$.
                \State Replace the edge $(u,v)$ by an edge $(s,v)$ of weight $d(w,u)+l(u,v)$ and label $w$.
            \EndIf
        \EndFor
        \State Run Dijkstra's algorithm from $s$.
        \For{each vertex $v$}
            \State $S_{i,j+1}(v)\gets S_{i,j}\cup\{\text{the label of the first edge of the shortest path to $v$}\}$.
        \EndFor
        \State \Return all $S_{i,j+1}(v)$.
    \end{algorithmic}
\end{algorithm}

\begin{lemma}[\cite{roditty2005deterministic} Theorem 2]
\label{lem:grow-all-balls}
    For any $j\ge 1$, \cref{alg:grow-ball} computes all $S_{i,j+1}(v)$ from all $S_{i,j}(v)$ in $O(m+n\log n)$ time by performing a single-source shortest paths computation on a graph with $O(n)$ vertices and $O(m)$ edges.
\end{lemma}

If all the balls have the same size and we want to grow all of them by $1$, we can call \cref{alg:grow-ball}. By \cref{lem:grow-all-balls}, the amortized time for each single grow operation will be $O(m/n+\log n)$, which matches with the randomized algorithm. However, in \cref{alg:hit-ball}, we only want to grow all unhit balls. Nevertheless, by observing some key properties, we can ``partially'' run \cref{alg:grow-ball} to achieve it.

\begin{lemma}
\label{lem:grow-unhit-balls}
    When running \cref{alg:hit-ball} for our current problem of selecting $A_{i+1}$ from $A_i$, each grow operation can be done in $O(m/n+\log n)$ amortized time.
\end{lemma}

\begin{proof}
    Note that during \cref{alg:hit-ball}, all the unhit balls have the same size, and we always grow them together. Suppose that there are $t$ unhit balls now, and they all have size $j$, and we want to grow all of them by $1$. Then we know $B_i(v)=S_{i,j}(v)$ for all unhit $B_i(v)$. Let's try to run \cref{alg:grow-ball}. Let $G'$ denote the graph constructed by \cref{alg:grow-ball} (which we won't actually construct).

    Notice that, if $B_i(u)$ is hit and $B_i(v)$ is unhit, then $S_{i,j}(u)\neq S_{i,j}(v)$ since $S_{i,j}(u)$ contains a center but $S_{i,j}(v)$ does not. So in $G'$, there are no edges between the vertices whose balls are hit and the vertices whose balls are unhit. Therefore, we can safely discard all vertices whose balls are hit. We only need to construct the partial graph for the $t$ vertices by examining edges incident to at least one of them.

    Consider each such edge $(u,v)$. If both $B_i(u)$ and $B_i(v)$ are unhit, then we just follow \cref{alg:grow-ball}. If $B_i(u)$ is hit and $B_i(v)$ is unhit, then we don't actually know $S_{i,j}(u)$. We only know $B_i(u)=S_{i,j'}(u)$ for some $j'\le j$. However, since $S_{i,j'}(u)$ will contain (at least) one center, which must not be in $S_{i,j}(v)$, we could still get the correct $w$ in line $6$.

    Since each vertex has degree $O(m/n)$, we only need to run Dijkstra's algorithm on a graph with $t+1$ vertices and $O(tm/n)$ edges, which can be done in $O(t(m/n+\log n))$ time as desired.
\end{proof}

Now, let's put everything together. We run \cref{alg:hit-ball} with $r=n^{-1/k}|A_i|$ and $p=1$. By \cref{thm:hitting-ball}, we have $|A_{i+1}|\le n^{-1/k}|A_i|$ and $\sum_v|B_i(v)|=O(n^{1+1/k})$. By \cref{lem:grow-unhit-balls}, the time we spend in this level is $O(n^{1+1/k}(m/n+\log n))=O(n^{1/k}(m+n\log n))$.
Note that in the last level $k-1$, since $A_k=\emptyset$ is fixed, we have $B_{k-1}(v)=A_{k-1}$. Since $|A_{k-1}|\le(n^{-1/k})^{k-1}|A_0|=n^{1/k}$, the above properties also hold for $i=k-1$.
Combining $k$ levels together, we construct the Thorup-Zwick oracle of size $O(kn^{1+1/k})$ in $O(kn^{1/k}(m+n\log n))$ deterministic time as desired, proving \cref{thm:do}.

\section*{Acknowledgments}

We thank Mikkel Thorup and the anonymous reviewers for helpful comments.

\bibliographystyle{alpha}
\bibliography{ref}



\end{document}